\documentclass[letterpaper,10pt,conference]{ieeeconf}
\IEEEoverridecommandlockouts
\usepackage{graphicx, subcaption} 
\graphicspath{{.},{./Images/}}

\usepackage{algorithm}
\usepackage{algpseudocode}
\usepackage{booktabs}
\usepackage[usenames,dvipsnames]{xcolor}
\usepackage{etoolbox}
\newtoggle{proofs}
\togglefalse{proofs}

\usepackage{amsthm}

\usepackage{svn-multi}
\svnidlong
{$LastChangedBy: frossi2 $}
{$LastChangedRevision: 5000 $}
{$LastChangedDate: 2014-10-03 12:29:48 -0500 (Ven, 03 Ott 2014) $}
{$HeadURL: https://segovia.stanford.edu/svn/asl/papers/2014/Rossi.Pavone.Allerton14/root.tex $}

\usepackage{color}
\usepackage{amsmath}
\usepackage{amssymb}
\usepackage{graphicx}
\usepackage{comment,xspace}
\usepackage{fancybox}

\newcommand{\real}{{\mathbb{R}}}
\newcommand{\reals}{\real}

\newtheorem{theorem}{Theorem}[section]
\newtheorem{proposition}[theorem]{Proposition}
\newtheorem{lemma}[theorem]{Lemma}

\title{Distributed consensus with mixed time/communication bandwidth
      performance metrics}
\author{Federico Rossi \thanks{Federico Rossi and Marco Pavone are  with the Department of Aeronautics and Astronautics, Stanford University, Stanford, CA, 94305, \{{\tt\small frossi2}, {\tt\small pavone}\}{\tt \small @stanford.edu}} \and Marco Pavone}

\begin{document}
\maketitle
\thispagestyle{empty}
\pagestyle{empty}

\begin{abstract}
In this paper we study the inherent trade-off between time and communication complexity for the distributed consensus problem. In our model, communication complexity is measured as the maximum data throughput (in bits per second) sent through the network at a given instant. Such a notion of communication complexity, referred to as bandwidth complexity, is related to the frequency bandwidth a designer should collectively allocate to the agents if they were to communicate via a wireless channel, which represents an important constraint for dense robotic networks. We prove  a lower bound on the bandwidth complexity of the consensus problem and provide a consensus algorithm that is bandwidth-optimal for a wide class of consensus functions. We then propose a distributed algorithm that can trade communication complexity versus time complexity as a function of a tunable parameter, which can be adjusted by a system designer as a function of the properties of the wireless communication channel.  We rigorously characterize the tunable algorithm's worst-case bandwidth complexity and show that it compares favorably with the bandwidth complexity of well-known consensus algorithm.
\end{abstract}

\section{Introduction}
Distributed decision-making in robotic networks is a ubiquitous problem, with applications as diverse as state estimation \cite{ROS:07}, formation control \cite{WR-RWB-EMA:07}, tracking \cite{YH-JH-LG:06} and cooperative task allocation \cite{MdW-BC:09}. Distributed decision-making problems such as leader election, majority voting, distributed  hypothesis testing, and some distributed optimization problems can all be abstracted as instances of the \emph{consensus} problem, where nodes in a robotic network have to agree on some common value \cite{JNT-MA:85,NL:96}. 
For this reason, the consensus problem has gathered significant interest in the Control Systems community in recent years, following the seminal works in \cite{JNT:84-extra, ROS-RMM:03c,AJ-JL-ASM:02}.

Research in the Control Systems community has mainly focused on the \emph{average} consensus problem and, in particular, on local averaging algorithms, where nodes repeatedly average their state with their neighbors': fundamental limitations on the \emph{time} complexity of local averaging algorithms are now known \cite{AO-JNT:07}.
Conversely, research in the Computer Science community has mainly focused on lower bounds on the general consensus problem in presence of node failures (both non-malicious and byzantine): several lower bounds on the time and communication complexity of the general consensus problem are now known. Significant attention has also been devoted to the  \emph{leader election} problem, a specific consensus problem where agents in a network select a single agent as their leader. Fundamental limitations of the leader election problem in terms of time and communication complexity and matching optimal algorithms are now well-understood \cite{NL:96}.
However, complexity results in the Control Systems and in the Computer Science communities strongly rely on the assumption that all messages can be delivered within a \emph{finite} amount of time that does \emph{not} depend on message size \cite{NL:96}. This assumption has significant effects on the \emph{frequency bandwidth} collectively required by the agents: however, despite the large interest in the consensus problem and its applications, the problem of bandwidth use has seen very limited investigation in both the Computer Science and the Control Systems communities. 

\emph{Motivation}:
In this paper we argue that bandwidth use can play a significant role in the real-world performance of consensus algorithms and significantly limit their scalability as the number of agents increases.

In order to appreciate the importance of bandwidth use in modern robotic networks, consider the following scenario. A network of $n=100$ robotic agents, each with six mechanical degrees of freedom is tasked with averaging their state (i.e. performing average consensus), e.g. to control their formation \cite{WR-RWB-EMA:07} or to filter noisy observations of a common target \cite{ROS:07}. Each agent's state or observation can be represented by twelve floating point numbers:  the size of each agent's initial condition is $b=768$ bits. Each message should be delivered in $10$ ms at most, in order to guarantee acceptable time performance. The complexity results in Section \ref{sec:algs} allow us to show that a bandwidth-optimal algorithm such as GHS with convergecast achieves consensus with a bandwidth complexity of 143 kbps (and converges in approximately 6.7 s). The popular average-based consensus algorithm requires 7750 kbps and has a significantly slower convergence rate, achieving convergence in approximately 100 s. Finally, the time-optimal flooding algorithm only requires 1 s to converge, but it requires a bandwidth of 775 Mbps.
As a comparison, a bandwidth-optimized protocol such as 802.11n WiFi requires use of the entire 2.4 GHz ISM band to transmit 288Mbps over very short distances: thus, selecting an inappropriate consensus algorithm can have a dramatic impact on the real-world performance, the scalability and potentially even the stability of a cyber-physical network. 

\emph{Statement of contributions}:
The contribution of this paper is threefold. First, we propose a rigorous metric for the bandwidth complexity of decentralized algorithms with omnidirectional (broadcast) communication channels and discuss its relevance to modern media access control (MAC) mechanisms. Second, we prove a lower bound on the bandwidth complexity of the generalized consensus problem. The bound is tight for a wide class of consensus functions that includes many distributed consensus problems such as mean and weighed mean, leader election, selected distributed optimization problems and majority voting. Finally, we show that the hybrid algorithm proposed by the authors in \cite{FR-MP:13} achieves intermediate bandwidth performance between the time-optimal flooding algorithm and the bandwidth-optimal GHS algorithm with convergecast. This result, combined with our previous findings, shows that our hybrid algorithm can be tuned to satisfy mixed time-bandwidth metrics as well as mixed time-energy metrics, trading time complexity (and, to an extent, robustness) for spectrum utilization and energy consumption. A graphical depiction of our lower bounds on time, byte and bandwidth complexity of the flooding algorithm, the GHS algorithm with convergecast and our hybrid algorithm is shown in Figure \ref{fig:bounds}.
\begin{figure*}[h!tp]
\centering
\begin{subfigure}[h]{0.3\textwidth}
\includegraphics[width=\textwidth]{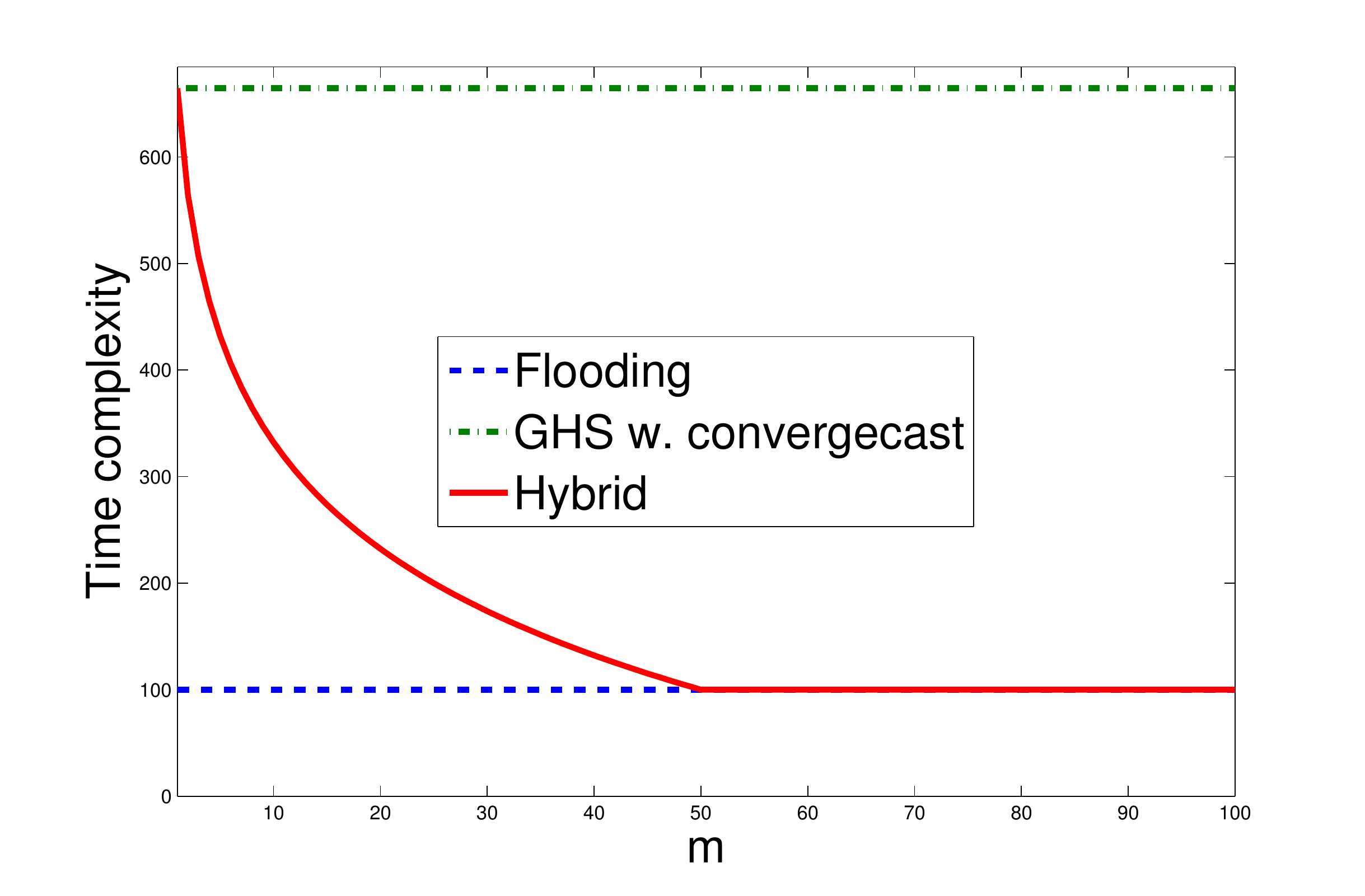}
\caption{Bounds on time complexity}
\end{subfigure}
\begin{subfigure}[h]{0.3\textwidth}
\includegraphics[width=\textwidth]{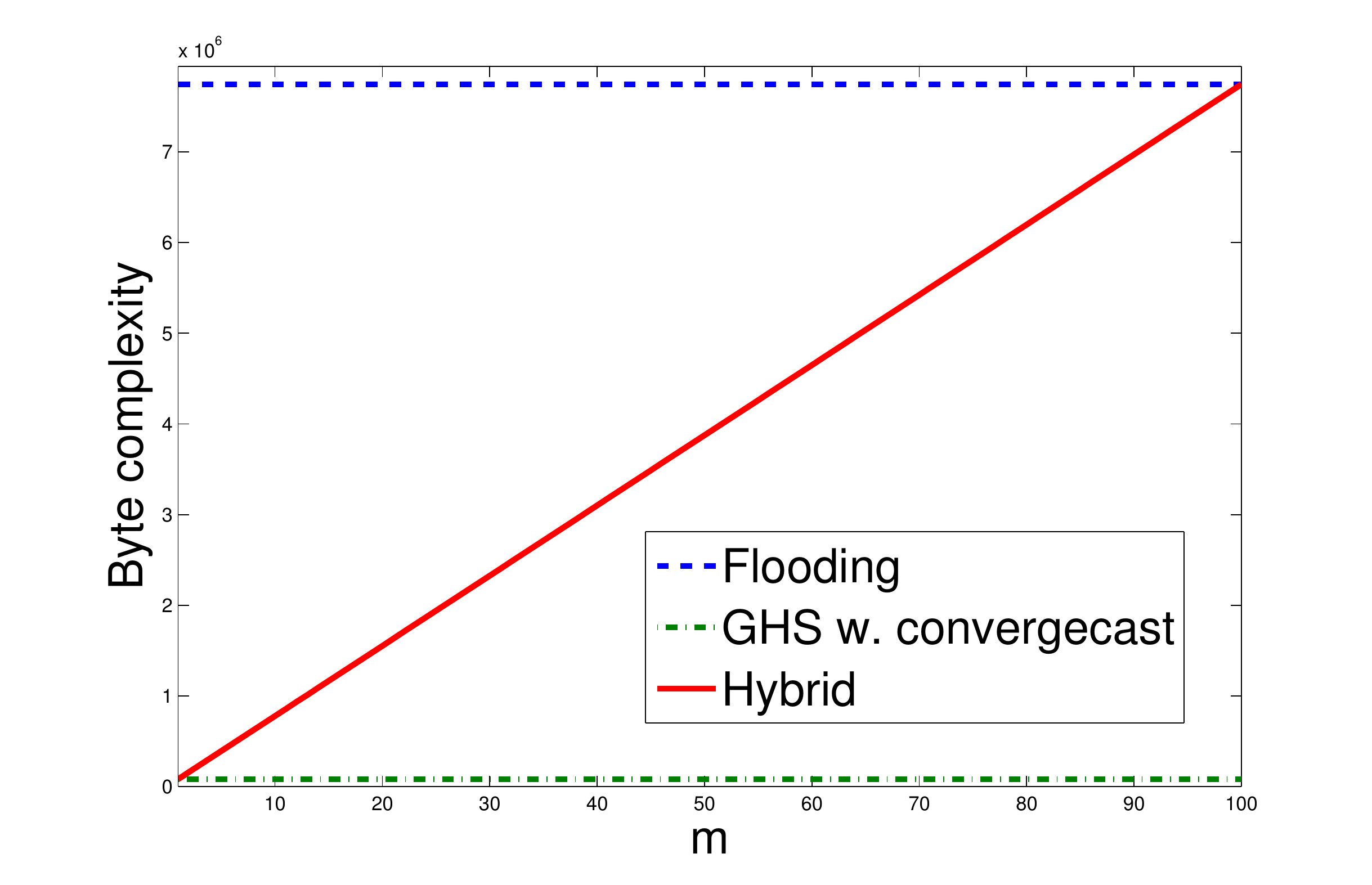}
\caption{Bounds on byte complexity}
\end{subfigure}
\begin{subfigure}[h]{0.3\textwidth}
\includegraphics[width=\textwidth]{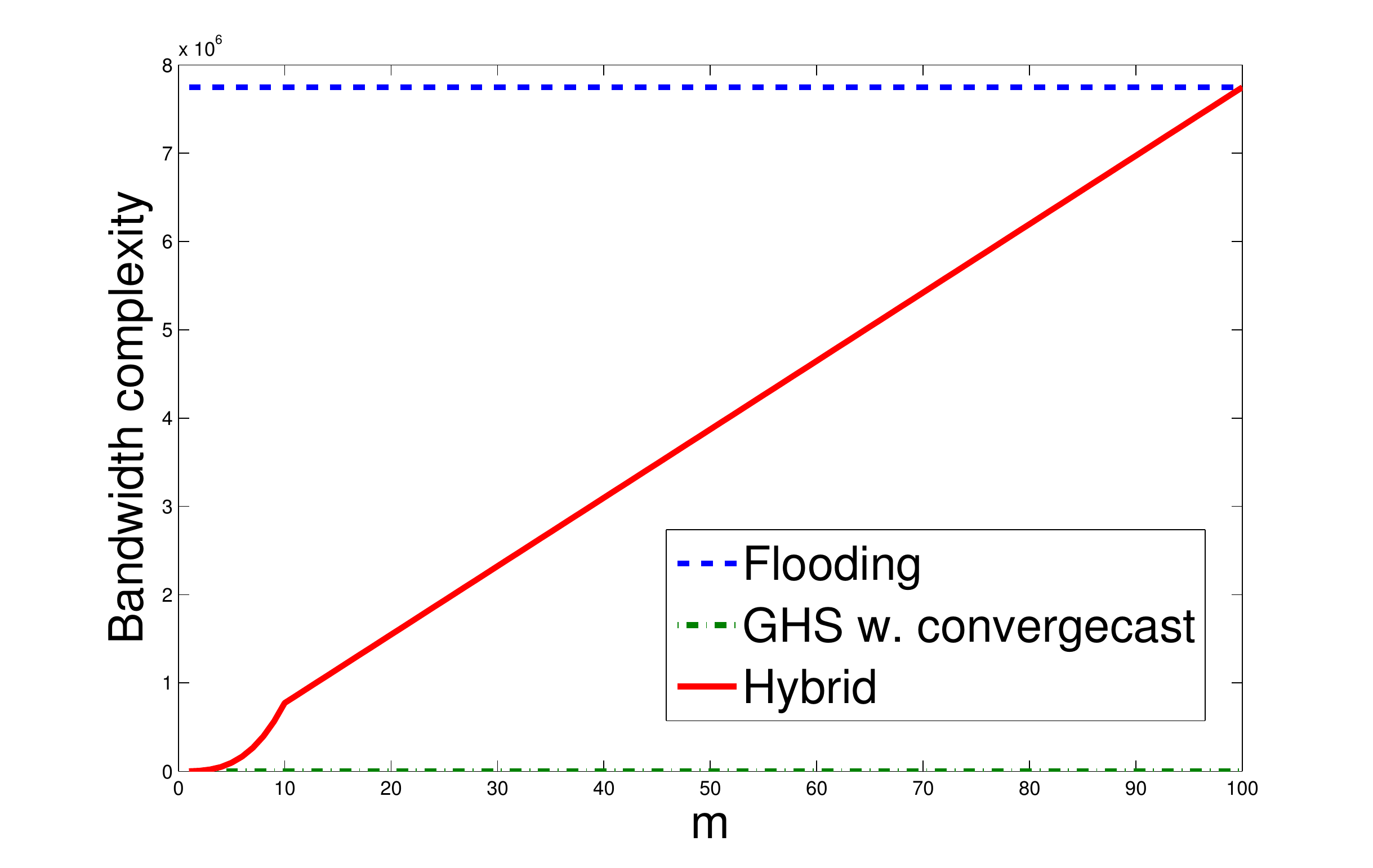}
\caption{Bounds on bandwidth complexity}
\end{subfigure}
\caption{Bounds on time, byte and bandwidth complexity of the flooding algorithm (in blue, dashed), GHS algorithm with convergecast (in green, dash-dotted) and of our hybrid algorithm (in red, solid) as a function of the tuning parameter $m$ for a network of size $n=100$ with state size $b=768$ bits.
Our hybrid algorithm recovers the time-optimal performance of the flooding algorithm as $m\to n$. Conversely, the hybrid algorithm recovers the byte and bandwidth performance of the byte-optimal and bandwidth-optimal GHS algorithm with convergecast, presented in Section \ref{sec:algs}, as $m\to 1$.}
\label{fig:bounds}
\end{figure*}

\emph{Organization}: This paper is organized as follows.
In Section \ref{sec:setup}, after formalizing our agent model and network model, we give a formal definition of bandwidth complexity and justify its relevance to modern multi-agent media access control mechanisms. We then provide a rigorous definition of the generalized consensus problem.
In Section \ref{sec:lbs} we prove a lower bound on the bandwidth complexity of the generalized consensus problem.
In Section \ref{sec:algs} we prove tightness of the lower bound. We also study the bandwidth complexity of common consensus algorithms (including the flooding algorithm, the GHS algorithm with convergecast and the average-based consensus-algorithm). A variation of the GHS algorithm with convergecast achieves the lower bound presented in the previous section.
In Section \ref{sec:tunable}, we study the bandwidth complexity of the hybrid algorithm introduced by the authors in \cite{FR-MP:13}. We  show that the  bandwidth complexity of the algorithm smoothly transitions from bandwidth-optimal performance to the performance of the time-optimal flooding algorithm.
Finally, in Section \ref{sec:conclusions}, we draw our conclusions and discuss directions for future research.

\section{Problem setup}
\label{sec:setup}
In this section we introduce the model and the complexity metrics used in this work. We also describe naming conventions that will be used in the rest of this paper. A preliminary version of the model has appeared in \cite{FR-MP:13}.
\subsection{Agent model}
An agent in a robotic network is modeled as an input/output (I/O) automaton, i.e., a labeled state transition system able to send messages, react to received messages and perform arbitrary internal transitions based on the current state and  messages received. 
A precise definition of I/O automaton is provided in \cite[pp. 200-204]{NL:96} and is omitted here in the interest of brevity. All agents in a robotic network are identical except for a unique identifier (UID - for example, an integer).  The time evolution of each agent  is characterized by two key assumptions:
\begin{itemize}
\item {\bf Fairness assumption}: the order in which transitions happen and messages are delivered is not fixed a priori. However, any enabled transition will \emph{eventually} happen and any sent message will \emph{eventually} be delivered.
\item {\bf Non-blocking assumption}: every transition is activated within $l$ time units of being enabled and every message is delivered within $d$ time units of being dispatched.
\end{itemize}
Essentially, the fairness assumption states that each  agent will eventually have an opportunity to perform transitions, while the non-blocking assumption gives timing guarantees (but no synchronization). We refer the interested reader to \cite[pages 212-215]{NL:96} for a detailed discussion of these assumptions. We argue here that these are \emph{minimal} assumptions for most  real-world robotic networks.

\subsection{Network model}
A \emph{robotic network} comprising $n$ agents is modeled as a \emph{connected}, \emph{undirected} graph $G = (V,E)$, where $V = \{1,\ldots, n\}$ is the node set, and $E\subset V\times V$, the edge set,  is a set of \emph{unordered} node pairs modeling the availability of a communication channel. Two nodes $i$ and $j$ are neighbors if $(i, j)\in E$.  The neighborhood set of node $i\in V$ is the set of nodes $j\in V$ that are neighbors of node $i$. Henceforth, we will refer to nodes and agents interchangeably. Our model is \emph{asynchronous}, i.e., computation steps within each node and communication are, in general, asynchronous.
In this paper we focus on \emph{static networks}, i.e., robotic networks where the edge set does not change during the execution of an algorithm. 
However, we do remark that (i) our lower bounds also apply to time-varying networks (although, of course, they may not be tight) and (ii) the hybrid algorithm described in section \ref{sec:tunable} has \emph{limited, tunable} resistance to network disruption, requiring a tunable amount of time to reconfigure after network disruptions. We refer the interested reader to \cite{FR-MP:13} for a thorough discussion of the recovery mechanism and of its time complexity.
\subsection{Model of communication}
Nodes communicate with their neighbors according to a \emph{local broadcast} communication scheme. A node can send a message to all neighbors simultaneously: the cost of a message (in terms of energy consumption and bandwidth use)  is independent of the number of receivers. 
We remark that local broadcast algorithms can emulate one-to-one communication: it is sufficient to append the intended recipient's UID to each broadcast and instruct non-recipients to ignore the message. 

The local broadcast communication scheme is representative of robotic networks where agents are equipped with \emph{omnidirectional} antennas: this arrangement is typical of most current airborne and ground-based robotic networks, where steerable antennas are unadvisable due to the agents' mobility. 

Message transmission requires a finite, nonvanishing amount of time; the non-blocking assumption ensures that every message is delivered within time $d$.
We consider $d$ to be constant: that is, all messages are delivered within a maximum time that does \emph{not} depend on message size or type. This assumption is widely used in the Computer Science community \cite{NL:96} and is typical of TCP-like communication protocols.
Thus, the parameter $d$ represents a \emph{desired performance level}.
\emph{Collisions} occur when (part of) two or more messages are transmitted on the same frequency at the same time: when a collision occurs, all messages involved in the collision are not delivered. When analyzing bandwidth complexity, we assume that messages are sent at a \emph{constant rate} throughout a window of length $d$: it is easy to see that ``bursty'' transmission would only decrease bandwidth performance.

We remark that, in presence of a TCP communication protocol, collisions do not cause messages to be permanently lost: nodes can sense the collision and resend the information at a later time. However, frequent collisions can have a major impact on the time required to deliver a message and, as a result, on the execution time of an algorithm.

We also remark that, in \emph{directional} communication schemes, communication channels can be \emph{spatially} separated: that is, two messages may be transmitted on the same frequency at the same time with minor interference if the respective directional communication channels do not physically overlap. In our omnidirectional communication scheme, on the other hand, collisions occur whenever two nodes within range of one another send messages on the same frequency at the same time; in addition, even if two nodes are not within range of one another, collisions may occur if they are trying to contact a third node in range of both (this problem is known as the ``hidden node problem'' in the telecommunication community). In this paper we assume no restrictions to the network topology the nodes can assume: we remark that there exist both dense and sparse network topologies where every node's communications may spatially interfere with all other nodes'.

\subsection{Bandwidth complexity measure}
We define bandwidth complexity  as the infimum worst-case (over graph topologies, initial values, fair executions and execution time) overall number of bytes transmitted at the same instant by all agents in the network.

Let $\mathcal G$ be a set of graphs with node set $V=\{1,\ldots, n\}$. For a given graph $G\in \mathcal G$, let $\mathcal F(a, x, G)$ be the set of \emph{fair executions}  for an algorithm  $a\in \mathcal A$ and a set of initial conditions $x \in \mathcal X^{n}$  (a fair execution is an execution of an algorithm that satisfies the fairness and non-blocking assumptions stated above).

Rigorously, the bandwidth complexity for a given consensus function $f$ with respect to the class of graphs $\mathcal G$ is
\begin{flalign*}
&\textrm{FC}(f, G):=&\\
& \inf_{a\in \mathcal A} \, \sup_{G\in \mathcal G}\, \sup_{x \in \mathcal X^{|G|}} \, \sup_{\alpha \in \mathcal F(a, x,G)} \, \sup_{t\in [0, T(a, x, \alpha, G)]} F(a, x, \alpha, G,t)
\end{flalign*}
where $F(a, x, \alpha, G,t)$ is the bandwidth (measured by the size of all messages transmitted at time $t$ divided by the maximum transmission time $d$) at time $t$ of execution $\alpha$ of algorithm $a$ with initial conditions $x$ on a graph $G$.
While very simple, the bandwidth complexity measure is a reliable proxy for many wireless communication protocols and media access control (MAC) mechanisms. Its interpretation varies depending on the specific MAC mechanism employed:
\begin{itemize}
\item If Frequency Division Multiple Access (FDMA) is employed, the frequency bandwidth required by a single message is proportional to the message size divided by the maximum transmission time $d$; the overall frequency bandwidth required by the network is proportional to the maximum \emph{overall} size of messages being transmitted at a given time divided by the maximum transmission time $d$, since every message must be broadcast on a different frequency slot;
\item If a Time Division Multiple Access (TDMA) media access control mechanism is employed, each agent is allocated a time slot in a round robin fashion so that only one agent can transmit during a given time slot. The sum of the durations of all time slots must be smaller than $d$ to guarantee that all sent messages be delivered within $d$ time units: thus, in order to guarantee a timely delivery, the frequency bandwidth required must be proportional to the maximum overall size of all messages sent at any instant of time, divided by $d$.
\item If Code Division Multiple Access (CDMA) is employed, multiple messages are relayed on the same, wide frequency spectrum at the same time; a \emph{spread spectrum} technique is employed to make decoding of sent messages possible. The bandwidth of the spread spectrum is significantly larger than the bandwidth of the uncoded signal: in particular (i) the bandwidth required by a single message before encoding is proportional to its size (in bytes) and (ii) the spreading gain is roughly proportional to the maximum number of users that the network can support. Thus, if all agents transmit messages of the same size at the same time (as will be the case for the proof of the lower bound we study in this paper and also for the algorithms we discuss in Section \ref{sec:algs}), bandwidth complexity captures the frequency spectrum required for successful communication with a CDMA MAC mechanism.

\item A rigorous study of the effect of available bandwidth on channel capacity when collision-detection mechanisms such as CSMA/CA are employed is beyond the scope of our work. However we remark that, for a given message size, increasing bandwidth reduces the time required to transmit a message and thus the network load, significantly reducing the probability of collisions and thus increasing the effective throughput.
\end{itemize}

Furthermore, regardless of the MAC mechanism employed, for large signal-to-noise ratios, the maximum capacity of a wireless channel is approximately proportional to bandwidth, as shown by Shannon in \cite{CES:49}.

\subsection{Model of computation}

In this paper we study collective decision-making problems where each node in the network is endowed with an initial value $x_i$ (which can be represented with $b$ bits) and \emph{each} node should output the value of a function of the initial values of \emph{all} nodes. In other words, each agent, after exchanging messages with its neighbors and performing internal state transitions, should output $f(x_1, \ldots, x_n)$ for some computable function $f$, which we call a \emph{consensus} function.
 We formalize the notions of consensus function as follows.

\textbf{Consensus functions}:
A consensus function is a \emph{computable} 
function $f : \mathcal X^n \mapsto \reals$ that depends on \emph{all} its arguments. More precisely, for each element $x = (x_1, \ldots, x_n) \in \mathcal X^n$ and for all $i\in\{1, \ldots, n\}$ one can find elements $x_{i}^{(1)}\in \mathcal X$ and $x_{i}^{(2)}\in \mathcal X$ such that
\[
f(x_1, \ldots, x_{i}^{(1)}, \ldots, x_n)\neq f(x_1, \ldots, x_{i}^{(2)}, \ldots, x_n).
\]
Loosely speaking, such choice of consensus function implies that each node is needed for the decision-making process. We collectively refer to problems involving the distributed computation of consensus functions (as defined above) as \emph{generalized consensus}. 

We introduce a \emph{representation} property for consensus functions that will be instrumental to derive fundamental limitations of performance in terms of the amount of information exchanged.

\textbf{Hierarchically computable consensus function}:
A consensus function is hierarchically computable if it can be written as  the composition of  a commutative and associative binary operator $\ast$, that is
\[
f(x_1, x_2,\ldots, x_n) = x_1 \ast x_2\ast \ldots\ast x_n.
\]
(The name is inspired by the observation that hierarchically computable functions can be computed with messages of small size on a \emph{hierarchical} structure such as a  tree). Furthermore, for a consensus function to be hierarchically computable, the consensus value as well as all intermediate products should be of the same size (in bytes) as the initial value. That is, if storing $x_i$ requires $\Theta(b)$ bytes, then storing the result of the operation $(x_i \ast x_j)$ and of the consensus value $f(x_1, \ldots, x_n)$ should also require $\Theta(b)$ bytes.

We remark that the class of hierarchically computable consensus functions includes average and weighed average, MAX and MIN (often used in leader election), voting and selected distributed optimization problems. We refer the reader to \cite{FR-MP:14a} for a more exhaustive characterization of this class of functions. 

\subsection{Nomenclature}
In the rest of this paper, we use the following definitions for nodes belonging to a rooted tree structure.
\begin{itemize}
\item Each rooted tree contains a \emph{root node}.
\item A node $j$ is the \emph{child} of node $i$ if (i) the shortest path from node $j$ to the root is one edge longer than the shortest path from node $i$ to the root and (ii) node $i$ and node $j$ share an edge in the tree. Conversely, node $i$ is called node $j$'s \emph{parent}.
\item A \emph{leaf node} is a node with no children.
\item A node $j$ is a \emph{descendant} of node $i$ if (i)  the shortest path from node $j$ to the root is strictly longer than the distance from node $i$ to the root and (ii) the path connecting node $j$ and the root node contains node $i$.
\item The set containing a node $j$ and all its descendants is the \emph{branch} of node $j$. 
\end{itemize}
Figure \ref{fig:treenomenclature} shows a graphical depiction of the definitions above.
\begin{figure}[h]
\centering
\includegraphics[width=.45\textwidth]{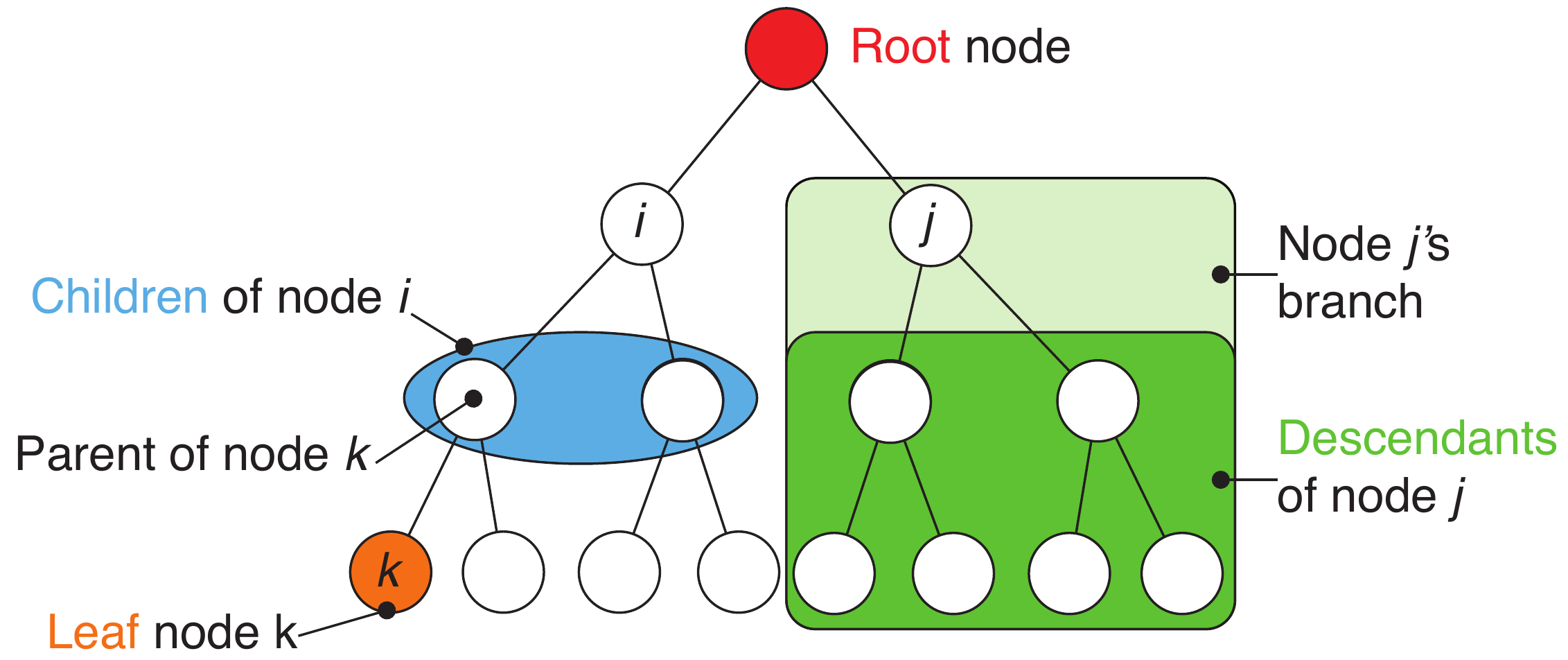}
\caption{Naming conventions for nodes belonging to a rooted tree structure}
\label{fig:treenomenclature}
\end{figure}

\section{Lower bounds on bandwidth complexity of consensus}
\label{sec:lbs}
In this section, we present two lower bounds on the bandwidth complexity of the consensus problem. The first bound applies to \emph{asynchronous} executions and only depends on very mild assumptions on the minimum number of available UIDs; the second bound applies to synchronous as well as asynchronous executions and requires slightly more restrictive assumptions on the execution time and on the minimum number of available UIDs. 

\begin{proposition}[Lower bound on the bandwidth complexity of the consensus problem in asynchronous executions]
\label{prop:asynclb}
Assume that messages carry a sender and/or a receiver ID. Assume also that agents' UIDs are selected from a set $S$ of cardinality $|S|\geq 2n$.  Then, for a given consensus function $f$ and class of graphs $\mathcal G$ with $n$ nodes,  $\textrm{FC}(f, \mathcal G) \in \Omega((n\log n +b)/d)$.
\end{proposition}
\begin{proof}

When they start execution, agents decide whether to send a message or whether to wait in silence until they hear a message. Since agents have no information about other nodes before they receive a message, their decision is based solely only on their UID and, possibly, their initial condition. We consider asynchronous executions: agents do not have access to a shared clock and therefore can not make decisions based on time.

We wish to show that, if the agents' UIDs are drawn from a set $U$ with cardinality $U\geq 2n$, then there exist a subset $M\subseteq U$ with cardinality $|M|\geq n$ such that all agents with an UID from $M$ send a message before receiving a message at least for one set of initial conditions.

We proceed by contradiction. Call $S=U\setminus M$ the set of UIDs such that, for all initial conditions, an agent sends no message before receiving one message. Now, assume by contradiction that the cardinality of $M$ is smaller than $n$. Then the cardinality of $S$ is larger than $n$: there exists an assignment of $n$ UIDs to the agents such that no node in the network sends a message before receiving one. Any network formed from these agents exchanges no messages and, in particular, it fails to solve the consensus problem unless all initial conditions are identical. We have reached a contradiction.

If the cardinality of $M$ is larger than $n$, then an adversary can select $n$ UIDs and $n$ matching initial conditions from $M$ and form a network where every node sends a message before receiving one. In an asynchronous executions, all nodes can transmit their first message simultaneously. Furthermore, at least $\log n$ bits are required to store $n$ distinct UIDs: thus, if messages contains the transmitter or the receiver UID, the size of each message is lower-bounded by $\log n$. Therefore $n\log n$ bits may be transmitted simultaneously, with a bandwidth complexity of $n\log n/d$. 

Finally, we observe that every agent transmits its initial value at least once. If this were not the case, then at least one agent would never inform other nodes of its initial value: then, since the consensus function is sensitive to all initial values, the algorithm would be unable to correctly compute the consensus function. Transmission of an initial value requires $(\log n+b)$ bytes: its bandwidth complexity is $(\log n + b)/d$.

We can then conclude that the bandwidth complexity of the consensus problem is $\textrm{FC}(f, \mathcal G) \in \Omega((n\log n +b)/d)$.
 \end{proof}
 
 Proposition \ref{prop:asynclb} strongly relies on the assumption of asynchronous communication. In the next proposition we show that the lower bound on bandwidth complexity also applies to synchronous executions if (i) the pool of UIDs is ``large enough'' and (ii) the algorithm is required to terminate in a bounded number of steps.
 
 \begin{proposition}[Lower bound on the bandwidth complexity of the consensus problem in synchronous and asynchronous executions]
\label{prop:synclb}
Assume that messages carry a sender and/or a receiver ID. Assume also that the consensus algorithm terminates within $R$ synchronous rounds and that  that agents' UIDs are selected from a set $S$ of cardinality $|S|\geq R (n+1)$.  Then, for a given consensus function $f$ and class of graphs $\mathcal G$ with $n$ nodes,  $\textrm{FC}(f, \mathcal G) \in \Omega((n\log n+b) /d)$.
\end{proposition}
\begin{proof}
In the synchronous setting, if an agent has received no messages from its neighbors, it decides whether to send a message or wait until it the next round based on (i) its UID $i$, (ii) its initial condition $x_i$ and (iii) the number of rounds $r$ elapsed since execution started. For each possible UID $i$, we call $\rho_i$ the \emph{smallest} round such that, for \emph{some} initial condition $x_{(i,\rho_i)}$, a node with UID $i$ and initial condition $x_{(i,\rho_i)}$ sends a message at round $\rho_i$ if it has not received a message until round $\rho_i-1$. Informally, $\rho_i$ is the first round when a node with UID $i$ may (for at least one initial condition) send a message if it hasn't received one; $\rho_i-1$ is the last round when a node with UID $i$ can not send a message unless it has received one, irrespective of its initial condition.
The number $\rho_i$ can only assume $R+1$ possible values (including $\infty$ if the node never sends a message before round $R+1$ unless it has received one). Furthermore, at most $n-1$ UIDs can have $\rho_i=\infty$: otherwise an adversary could build a network where no agent sends a message before round $R+1$, and therefore consensus is not achieved is achieved within $R$ rounds.

Then, by the pigeonhole principle, there exists a round $\bar \rho$ with $1\leq \bar\rho \leq R$ such that at least $n$ UIDs have $\rho_i=\bar \rho$. An adversary can arrange agents with these UIDs (and matching initial conditions) in a network: then all $n$ agents will send a message at round $\bar \rho$, after staying silent for the first $\bar \rho -1$ rounds. Since every message carries the transmitter or the receiver UID, the size of each message is lower-bounded by $\log n$. Thus,  $n\log n$ bits may be sent at round $\bar\rho$.

Our argument is completed by the observation that, as in Proposition \ref{prop:asynclb}, every agent sends its initial value at least once: this operation has a bandwidth complexity of $\Omega((\log n + b)/d)$.

Thus, the broadcast complexity of the consensus problem in the synchronous case is $\textrm{FC}(f, \mathcal G) \in \Omega((n\log n +b)/d)$.
\end{proof}

We have now proven a lower bound on the bandwidth complexity of the consensus problem for synchronous and asynchronous executions. In the following section we show tightness of this bound by proposing a bandwidth-optimal algorithm, then we compare its performance with commonly-used consensus procedures.

\section{Tightness of the lower bound and bandwidth complexity of common consensus algorithms}
\label{sec:algs}
In this section, we analyze the byte complexity of the average-based consensus algorithm \cite{JNT:84-extra,ROS-RMM:03c,AJ-JL-ASM:02}, the flooding consensus algorithm \cite{NL:96} and the GHS algorithm with convergecast \cite{RGG-PAH-PMS:83,FR-MP:13}. We also propose a modified version of the GHS algorithm that achieves the lower bound presented in Propositions \ref{prop:asynclb} and \ref{prop:synclb}.

\subsection{Bandwidth complexity of common consensus algorithms}
\begin{lemma}[Bandwidth complexity of the average-based consensus algorithm]
\label{prop:avgbased}
Assume that messages carry a sender and/or a receiver UID. Then the bandwidth complexity of the average-based consensus algorithm is $O(n(\log n + b)/d)$, where $b$ is the size of an agent's initial condition.
\end{lemma}
\begin{proof}
In the average-based consensus algorithm, nodes maintain a local estimate of the consensus value (which has the same type and size as the nodes' initial value). At each time step, nodes send their estimate to their neighbor, then update their estimate as the \emph{average} of their estimate and their neighbors'. Thus, if messages carry the sender UID, each node transmits $\log n +b$ bits with each message. All agents may communicate at once: this is the case, for instance, in synchronous executions, which are a special case of more general asynchronous executions. The resulting bandwidth complexity is $O(n(\log n + b)/d)$.  
\end{proof}
\begin{lemma}[Bandwidth complexity of the flooding consensus algorithm]
\label{prop:flooding}
Assume that messages carry a sender and/or a receiver UID. Then the bandwidth complexity of the flooding consensus algorithm is $O(n^2(\log n + b)/d)$, where $b$ is the size of an agent's initial condition.
\end{lemma}
\begin{proof}
In a flooding algorithm, at each time step, every node sends to all neighbors all \emph{new} information it has received at the previous (asynchronous) round. It is easy to observe that, for certain network topologies (e.g. for the complete graph) every node may receive information from $O(n)$ other nodes at the same time and thus retransmit information from $O(n)$ nodes simultaneously. Each piece of information from one node has size $\log n + b$: thus, all nodes may send messages of size $n(\log n + b)$. Furthermore, all nodes may send large messages simultaneously (e.g., in a synchronous execution). Thus the bandwidth complexity of the flooding algorithm is $O(n^2(\log n + b)/d)$. 
\end{proof}
\begin{lemma}[Bandwidth complexity of the GHS algorithm with convergecast]
\label{prop:GHSub}
Assume that messages carry a sender and/or a receiver UID. Assume also that the consensus function is hierarchically computable. Then the bandwidth complexity of the GHS algorithm with convergecast is $O((n\log n+nb)/d)$, where $b$ is the size of an agent's initial condition.
\end{lemma}
\begin{proof}
The GHS algorithm builds a rooted minimum spanning tree by repeatedly merging non-spanning trees across (minimum-weight) edges. The algorithm requires each edge to have a unique weight, while our model considers an unweighed graph: thus, we assign to each edge a unique, arbitrary, weight\footnote{One popular choice for edge weights is to assign to each edge a ``weight'' equal to the UIDs of the two nodes incident on the edge and use a lexicographic ordering.}. Messages exchanged by the agents during execution only contain node IDs, edge weights and boolean values: thus, the size of each message is upper-bounded by $O(\log n)$. Since no message is larger than $\log n$ and nodes only send one (broadcast) message at any given time, the bandwidth complexity of the GHS algorithm is $O(n\log n /d)$. We refer the interested reader to \cite{RGG-PAH-PMS:83} for an in-depth discussion of the algorithm. 

Once a rooted tree has been established, the root contacts all other nodes via a tree broadcast and asks them to relay their values. Nodes then relay their values through a tree convergecast: every node waits until it has heard back from all its children (if any), then computes the consensus value for its branch (this is always possible if the consensus function is hierarchically computable) and relays it to its parent. When the root learns every child's consensus value, it computes the overall consensus value and relays it to every node via a tree broadcast.

All messages exchanged in this phase have size $O(\log n + b)$: every message contains a consensus value of size $b$ and a sender and/or a receiver ID. There exist tree network topologies (e.g. a tree of depth one) where up to $n-1$ nodes may send messages simultaneously during a tree convergecast: thus, the bandwidth complexity of this phase is $O(n(b+\log n)/d)$. This completes our proof.
\end{proof}

We remark in passing that  Awerbuch's minimum spanning tree algorithm \cite{BAw:87}, which improves the GHS algorithm's time complexity to a time-optimal $O(n)$, is \emph{not} bandwidth-optimal: in particular, the \emph{Test-Distance} procedure has a bandwidth complexity of $O(n\log^2 n)$ during certain executions on selected network topologies, since each Test-Distance message needs to keep track of the return route to its sender.

\subsection{Tightness of the lower bound on bandwidth complexity}
The GHS-inspired consensus algorithm outlined above is not bandwidth-optimal: optimality, however, can be achieved with a simple modification that does not influence asymptotic time, message or byte complexity\footnote{Informally, time complexity is the time required by the algorithm to converge, message complexity is the overall number of messages exchanged by all agents and byte complexity is the overall number of bytes exchanged by all agents. We refer the reader to \cite{FR-MP:13} and \cite{FR-MP:14a} for a rigorous definition of these complexity metrics.}.

The tree-building phase is unchanged. When computing the nodes's consensus function, we exploit the tree structure to make sure that only one node sends a message at any given time. Specifically, once a rooted tree has been established, every node contacts its children \emph{one by one}: children are ordered arbitrarily and every child node is only contacted once the previous child has returned its branch's consensus function.
Pseudocode for the algorithm is reported in Algorithm \ref{alg:slowconvergecast}.

\begin{algorithm}
\caption{Bandwidth-optimal consensus function computation on a tree} \label{alg:slowconvergecast}
\begin{algorithmic}
\floatname{algorithm}{Procedure}
\renewcommand{\algorithmicrequire}{\textbf{Input:}}
\renewcommand{\algorithmicensure}{\textbf{Output:}}

\Require $IsRoot$ \Comment{a bool}
\State $Children$ \Comment{An ordered list of child nodes}
\State $Parent$ \Comment{The parent node's ID}
\State $x$ \Comment{The node's initial condition}
\State $c$ \Comment{A local estimate of the consensus value}

\If{$IsRoot$ is true}
\State\Call{ComputeConsensusValue}{}\Comment{The root initializes the consensus procedure}
\EndIf

\Procedure{ComputeConsensusValue}{}
\ForAll{$Child$ in $Children$}
	\State Ask $Child$ to \Call{ComputeConsensusValue}{}
	\State Wait for $Child's$ \Call{ConsensusReply}{$c_{Child}$}.
	\State $c \gets$ \Call{UpdateConsensus}{$c$,$c_{Child}$}
\EndFor
\If{$IsRoot$ is true} \Comment{Inform every node of the consensus value}
	\State Ask $Child$ to \Call{RelayConsensusValue}{c}
	\State Wait for $Child's$ \Call{AckConsensus}{}.
\Else
\State Send $Parent$ a \Call{ConsensusReply}{$c$}\Comment{Inform parent of the branch's consensus value}
\EndIf
\EndProcedure

\Procedure{RelayConsensusValue}{$c_r$}
\State $c\gets c_r$\Comment{Store the consensus value}
\ForAll{$Child$ in $Children$}
	\State Ask $Child$ to \Call{RelayConsensusValue}{c}
	\State Wait for $Child's$ \Call{AckConsensus}{}.
\EndFor
\If{$IsRoot$ is false}
\State Send $Parent$ an \Call{AckConsensus}{}
\EndIf
\State Terminate
\EndProcedure

\Function{UpdateConsensus}{$c_i,x_j$}
\State $x_i\gets c_i * x_j$
\EndFunction

\end{algorithmic}
\end{algorithm}

It is easy to prove that the bandwidth complexity of this consensus function computation on a tree is $O((\log n + b)/d)$.
\begin{lemma}[Bandwidth complexity of Algorithm \ref{alg:slowconvergecast}]
\label{prop:ccfc}
Assume that messages carry the sender and/or the receiver ID. Assume also that the consensus function is hierarchically computable. Then the bandwidth complexity of Algorithm \ref{alg:slowconvergecast} is $O((\log n + b)/d)$. 
\end{lemma}
\begin{proof}
The algorithm is, essentially, a token-passing algorithm. The token originates at the root. Nodes, starting with the root, pass the token to their children, one at a time, when they ask each child to compute its consensus value; children pass the token back to their parent when they relay the local estimate of the consensus value. Nodes (starting at the root) then pass the token to their children, one at a time, when they relay the consensus value; children return the token to their parent when they acknowledge reception of the consensus value. It is easy to see that (i) a node only sends a message when it holds the token and (ii) the token is never duplicated, since each node only contacts one child when it has heard back from the previous child and each node only has one parent. Thus, only one node can send a message at any given time. Messages contain a local estimate of the consensus value, a sender ID and a receiver ID: their size is $b+2\log n$. The claim follows.
\end{proof}

Lemma \ref{prop:ccfc} allows us to prove tightness of the lower bounds on bandwidth complexity.
\begin{proposition}[Bandwidth complexity of the consensus problem in asynchronous executions]
Assume that messages carry the sender and/or the receiver ID. Assume also that agents' UIDs are selected from a set $S$ of cardinality $|S|\geq 2n$ and that the consensus function is hierarchically computable. 
Then, for a given consensus function $f$ and class of graphs $\mathcal G$ with $n$ nodes,  $\textrm{FC}(f, \mathcal G) \in \Theta((n\log n+b) /d)$.
\end{proposition}
\begin{proof}
In order to prove the claim, we need to find an algorithm with bandwidth complexity $O((n\log n +b)/d)$. The GHS algorithm has a bandwidth complexity of $O(n\log n)$, as shown in the first part of Lemma \ref{prop:GHSub}. Once a tree structure has been established, Algorithm \ref{alg:slowconvergecast} computes the consensus function with a bandwidth complexity of $O((\log n + b)/d)$. The claim then follows from Proposition \ref{prop:asynclb}.
\end{proof}

The proof of the next proposition is identical to the one above and is therefore omitted.

\begin{proposition}[Bandwidth complexity of the consensus problem in synchronous and asynchronous executions]
Assume that messages carry a sender and/or a receiver ID. Assume also that the consensus algorithm terminates within $R$ synchronous rounds and that  that agents' UIDs are selected from a set $S$ of cardinality $|S|\geq R (n+1)$.  Then, for a given consensus function $f$ and class of graphs $\mathcal G$ with $n$ nodes,  $\textrm{FC}(f, \mathcal G) \in \Theta((n\log n+b) /d)$.
\end{proposition}

We note in passing that Algorithm \ref{alg:slowconvergecast} has the same worst-case time, message and byte complexity as the convergecast algorithm proposed in \cite{FR-MP:13}, as shown in the two lemmas below.
\begin{lemma}[Time complexity of Algorithm \ref{alg:slowconvergecast}]
\label{prop:cctc}
The time complexity of Algorithm \ref{alg:slowconvergecast} is $O(n)$.
\end{lemma}
\begin{proof}
Every node except for the root receives two messages from its parent (one asking to compute its branch's consensus value and one relaying the network's consensus value) and sends two messages to its parent (one informing the parent of the branch's consensus value and one to acknowledge reception of the network's consensus value). Only one message is sent at any given time,as shown in Proposition \ref{prop:ccfc} . Thus, the time complexity of Algorithm \ref{alg:slowconvergecast} is $4(n-1)$. 
\end{proof}
\begin{lemma}[Message and byte complexity of Algorithm \ref{alg:slowconvergecast}]
Assume that messages carry the sender and/or the receiver ID. Assume also that the consensus function is hierarchically computable. Then the message and byte complexity of Algorithm \ref{alg:slowconvergecast} are $O(n)$ and $O(n(\log n + b)$ respectively. 
\end{lemma}
\begin{proof}
As shown in Proposition \ref{prop:cctc}, $4(n-1)$ messages are sent in the network. Furthermore, if the consensus function is hierarchically computable, each message has size $(\log n + b)$, as shown in Proposition \ref{prop:ccfc}. The claim follows.
\end{proof}

\section{A tunable algorithm}
\label{sec:tunable}
The algorithms presented in Section \ref{sec:algs} offer optimal performance with respect to different complexity metrics.\

The flooding algorithm is time-optimal for any graph (and not only worst-case optimal over the class $\mathcal G$ of graphs with $n$ nodes); furthermore, it offers \emph{maximal} robustness to disruptions of a communication channel and performs well on time-varying networks. On the other hand, the bandwidth complexity of flooding is the worst among the algorithms we study and its byte complexity (discussed in \cite{FR-MP:13}) is also very suboptimal.

The average-based algorithm performs no better than flooding with respect to byte complexity and robustness, and it has significantly worse time complexity. On the other hand, its bandwidth complexity is better than the flooding algorithm's, but it is not optimal.

The GHS-inspired algorithm has optimal bandwidth and byte complexity. However, the algorithm has minimal robustness margins to the disruption of a communication channel.

In this section we show how the hybrid, tunable algorithm proposed by the authors in \cite{FR-MP:13} achieves \emph{intermediate} bandwidth performance between the flooding algorithm and the GHS-inspired algorithm, recovering the time-optimal behavior of the flooding algorithm and the bandwidth-optimal behavior of GHS for different values of the tuning parameter. This result complements our previous findings on the time and byte complexity of the hybrid algorithm and shows how this algorithm can achieve \emph{mixed} time/bandwidth performance metrics. 

Due to space limitations, we only report a high-level description of the algorithm: we refer the interested reader to \cite{FR-MP:13} for details\footnote{We remark two very minor changes with respect to the detailed description in \cite{FR-MP:13}: we replace the challenge-response exchanges at the end of Phase 1 and at the beginning of Phase 3 with broadcasts, to better exploit the \emph{local broadcast} communication model; furthermore, we do not duplicate messages exchanged between clusters for simplicity.}.

Our algorithm operates in four phases. 
Phase 1 starts by building a forest of \emph{minimum weight} trees (shown in Figure \ref{fig:hybridalgphase1}) of height $O(n/m)$. 
All nodes run a modified version of the GHS algorithm \cite{RGG-PAH-PMS:83} which only differs from the original algorithm with respect to the stopping criterion. When a node discovers that Phase 1 is over, it informs its neighbors with a broadcast. When a node has finished Phase 1 and all its neighbors have, too, it enters Phase 2.

 In Phase 2, tree height is upper-bounded by splitting clusters with too many agents while enforcing a lower bound on tree size. This phase of the algorithm starts at the leaves of each tree. Each node recursively counts the number of its descendants moving towards the root; agents with more than $\lfloor n/m \rfloor$ offspring create a new cluster, of which they become the root, and cut the connection with their fathers. The tree containing the original root may be left with too few nodes: the root can undo one cut to counter this.

In Phase 3, each tree establishes connections with neighbor clusters, as shown in Figure \ref{fig:hybridalgphase3}. When a node switches to Phase 3, it informs all neighbors of its Cluster ID with a broadcast. The root of each cluster is then informed of the available connections to neighbor clusters with a convergecast starting at the leaves. Specifically, as soon as a node knows (i) what clusters its children are connected to (either directly or through their children) and (ii) each neighbor agent's cluster ID, it informs its parent about which clusters it is connected to (either directly or indirectly).

Roots also exploit the tree structure to compute their cluster's consensus function via Algorithm \ref{alg:slowconvergecast}: once they have computed the cluster's consensus function, they switch to Phase 4.

In Phase 4, cluster roots communicate with each other through the connections discovered in the previous stage, as shown in Figure \ref{fig:hybridalgphase4}. Conceptually, this phase of the algorithm is simply flooding across clusters. Each root sends a message containing its cluster's consensus function to each neighbor tree through the connections built in Phase 3.
When a root learns new information, it forwards it once to its neighbor clusters via the same mechanism.

If a link failure breaks one of the trees (as in Figure \ref{fig:hybridalgphaseF}), the two halves evaluate their size. If either of the two halves is too small, it rejoins an existing cluster; a splitting procedure guarantees that tree height stays bounded. After failure, all nodes update their routing tables.

Finally, when a link outside a tree fails, nodes on the two sides of the failure update their routing tables and notify their cluster roots.

\begin{figure}[h]
\centering
\begin{subfigure}[b]{0.2\textwidth}
\includegraphics[width=\textwidth]{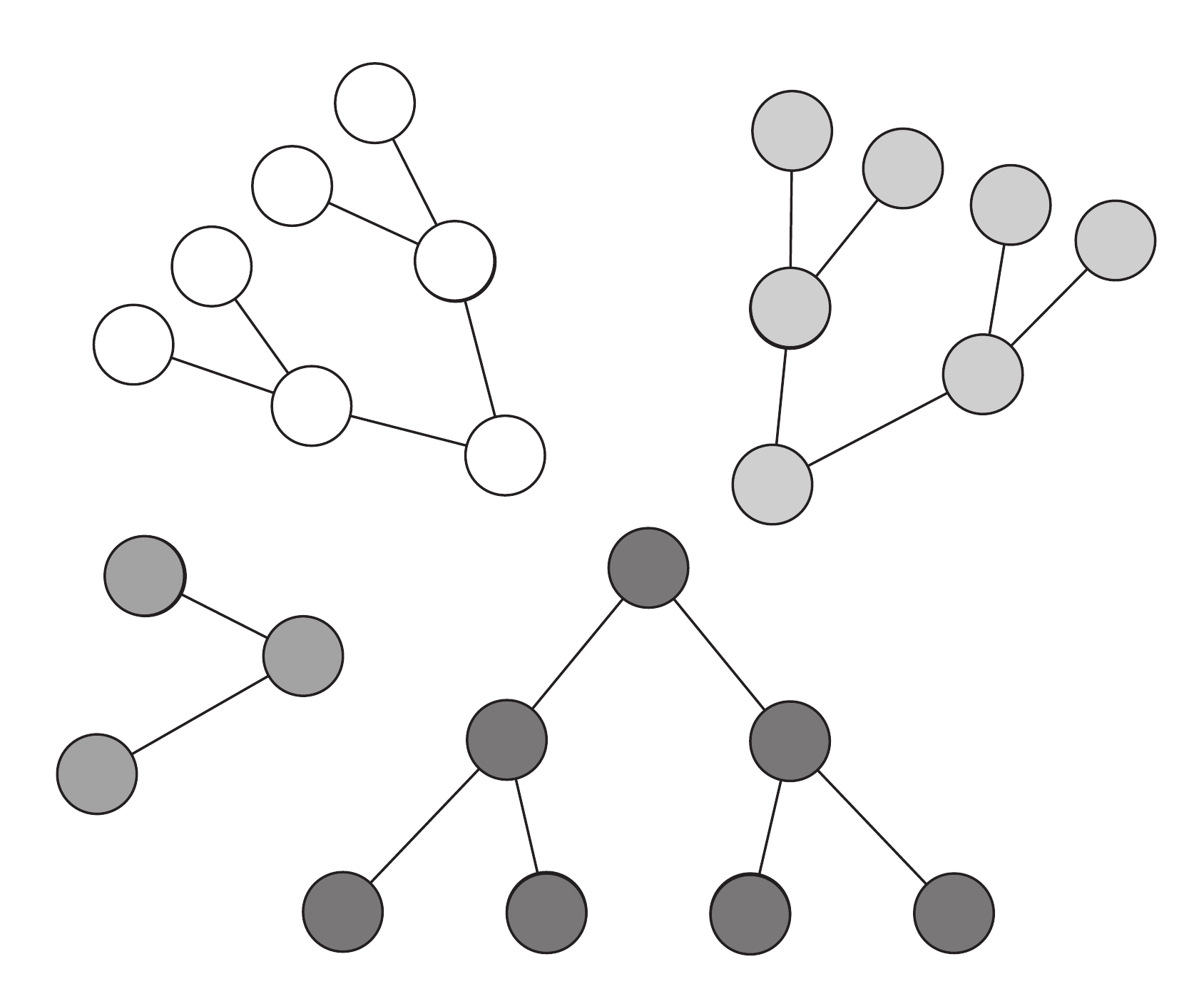}
\caption{Phase 1 and 2: forest building}
\label{fig:hybridalgphase1}
\end{subfigure}
\begin{subfigure}[b]{0.2\textwidth}
\includegraphics[width=\textwidth]{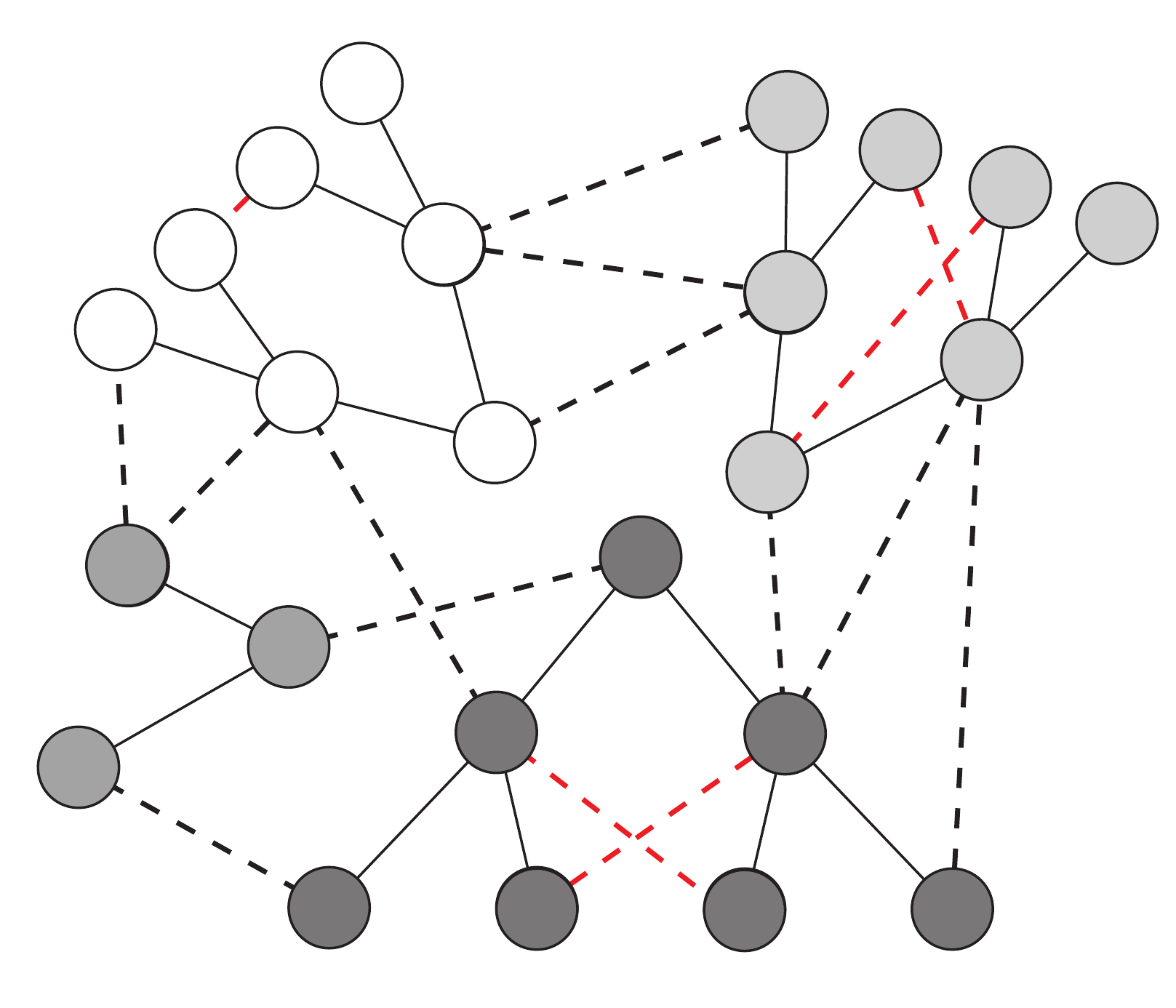}
\caption{Phase 3: establishment of inter-cluster links}
\label{fig:hybridalgphase3}
\end{subfigure}

\begin{subfigure}[b]{0.2\textwidth}
\includegraphics[width=\textwidth]{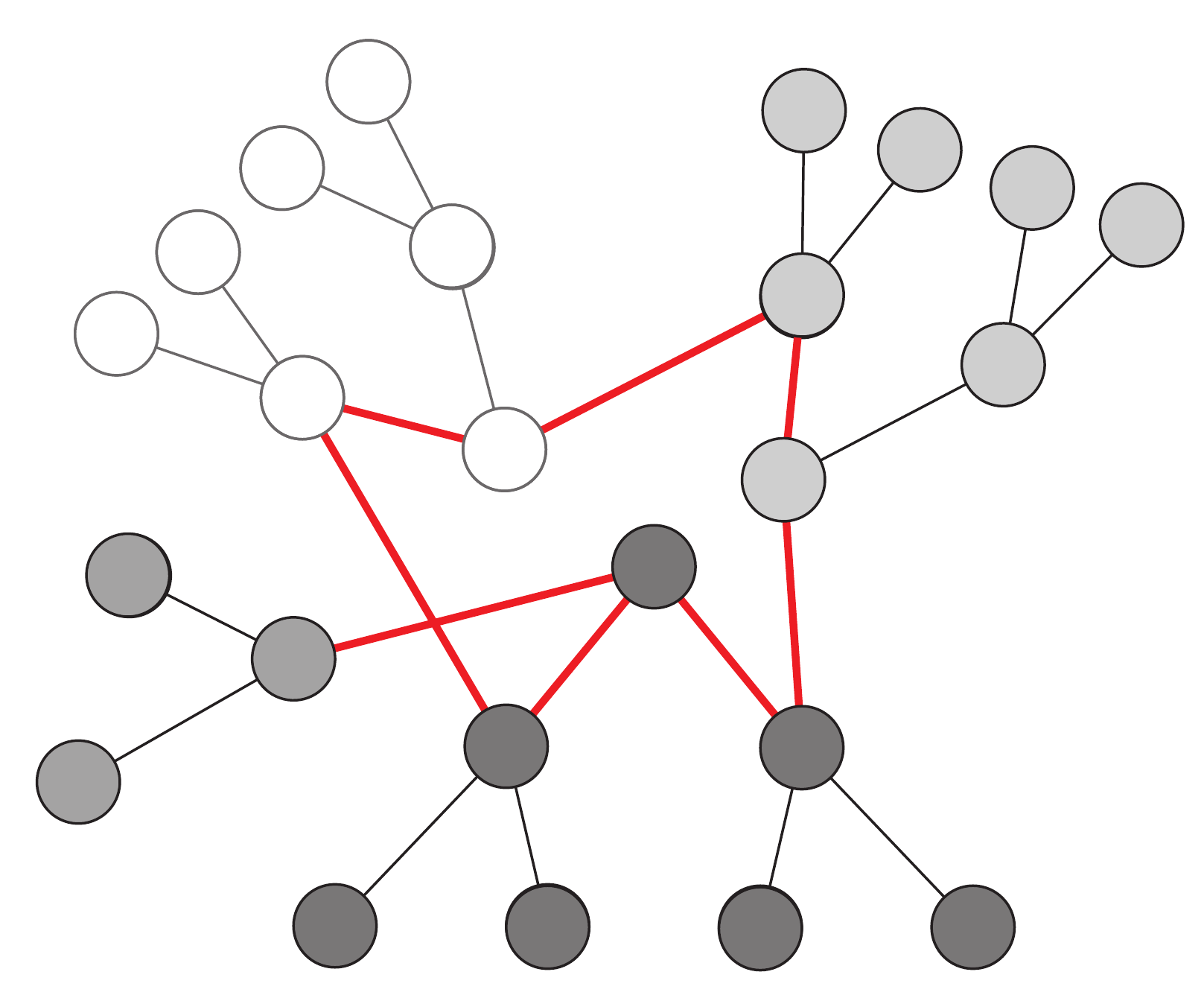}
\caption{Phase 4: inter-cluster flooding}
\label{fig:hybridalgphase4}
\end{subfigure}
\begin{subfigure}[b]{0.2\textwidth}
\includegraphics[width=\textwidth]{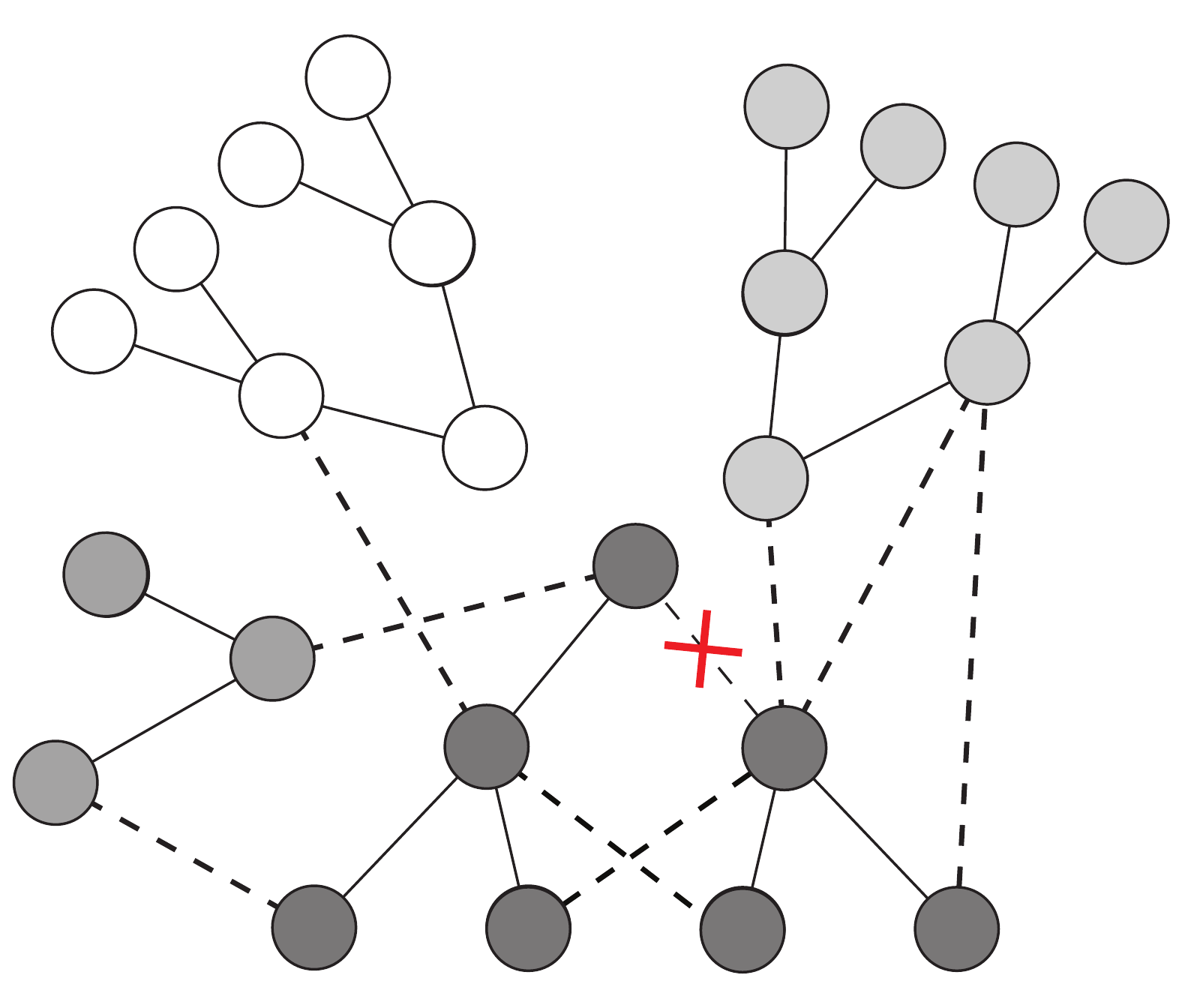}
\caption{Phase F: recovery from failure}
\label{fig:hybridalgphaseF}
\end{subfigure}
\caption{Schematic representation of the hybrid algorithm behavior.}
\label{fig:hybridalgschematic}
\end{figure}

We now study the bandwidth complexity of the tunable consensus algorithm.

\begin{lemma}[Bandwidth complexity of Phase 1 of the tunable algorithm]
\label{lemma:hybrid1ub}
Assume that messages carry the sender UID. Assume also that the consensus function is hierarchically computable. Then the bandwidth complexity of Phase 1 of the tunable algorithm is $O(n\log n/d)$.
\end{lemma}
\begin{proof}
Phase 1 only differs from the GHS algorithm, whose bandwidth complexity is shown to be $O(n\log n/d)$ in Lemma \ref{prop:GHSub}, with respect to the stopping criterion. When a node stops, it informs its neighbors: the size of the message informing the neighbors is $O(\log n)$, since it only contains the node's UID and a constant-size message. All nodes may stop and inform their neighbors at once: thus the bandwidth complexity of Phase 1 is $O(n\log n /d)$.
\end{proof}

\begin{lemma}[Bandwidth complexity of Phase 2 of the tunable algorithm]
\label{lemma:hybrid2ub}
Assume that messages carry the sender UID. Assume also that the consensus function is hierarchically computable. Then the bandwidth complexity of Phase 2 of the tunable algorithm is $O(n\log n/d)$.
\end{lemma}
\begin{proof}
The size of all messages exchanged in Phase 2 is $O(\log n)$: each message contains the number of agents that are descendants of the sender node and a cut ID containing two UIDs. Every node has at most $n-1$ descendants: thus, the number of descendants of a given node can be represented with $O(\log n)$ bits. Since every node sends at most one message at a time, the bandwidth complexity of Phase 2 is $O(n\log n/d)$.
\end{proof}

\begin{lemma}[Bandwidth complexity of Phase 3 of the tunable algorithm]
\label{lemma:hybrid3ub}
Assume that messages carry the sender UID. Assume also that the consensus function is hierarchically computable. Then the bandwidth complexity of Phase 3 of the algorithm is $O((nm\log n+b)/d)$.
\end{lemma}
\begin{proof}
Two types of messages are exchanged (potentially at the same time) in the cluster discovery routine in Phase 3. First, nodes announce their cluster ID: each message has size $\log n$ (since a cluster ID is the ID of the root node) and every node sends exactly one such message. Then, nodes send their parents a list of the clusters their branch is connected to. Every message contains up to $m$ cluster IDs, each of size $\log n$: thus, the size of each message is upper-bounded by $m\log n$. Every node sends at most one message at any given time: thus the bandwidth complexity of Phase 3 is $O(nm\log n/d)$.

In addition, each root computes the cluster's consensus function with Algorithm \ref{alg:slowconvergecast}: the bandwidth complexity of this operation is $O((\log n +b)/d)$, as shown in Lemma \ref{prop:ccfc}.

Thus, the overall byte complexity of Phase 3 of the algorithm is $O((nm\log n+b)/d)$.
\end{proof}

\begin{lemma}[Bandwidth complexity of Phase 4 of the tunable algorithm]
\label{lemma:hybrid4ub}
Assume that messages carry the sender UID. Assume also that the consensus function is hierarchically computable. Then the bandwidth complexity of Phase 3 of the algorithm is $O(nm(b+\log n)/d)$ and $O(m^3(b+\log n)/d)$ .
\end{lemma}
\begin{proof}
We prove the two bounds separately. First, we note that that every message exchanged in Phase 4 contains a list of at most $m-1$ intermediate consensus values (one per cluster) and the ID of the relevant cluster: thus, the size of each message is $O(m(\log n + b))$. Since no node sends more than a message at a time, the first bound follows.

Second, we observe that in the cluster flooding algorithm every cluster forwards each new piece of information it receives from a (neighbor or non-neighbor) cluster to neighbor clusters exactly once. Thus, each cluster contacts each of the $O(m-1)$ neighbor clusters with at most $m(b+\log n)$ bits of information overall.

Each exchange between neighbor clusters may require multiple messages. However, it is easy to see that only one message per exchange is transmitted at a given time: the message is \emph{routed} from the root of a cluster to the root of its neighbor, thanks to the routing tables developed in Phase 3, with no duplication. Thus, even if all $m$ clusters send information about all $m-1$ other clusters to every neighbor at the same time, no more than $O(m^3(\log n + b)/d)$ bits are exchanged at any given time.

Once a cluster root has computed the overall consensus value, it informs all nodes in its cluster. This is done with Algorithm \ref{alg:slowconvergecast}, with bandwidth complexity $O((\log n + b)/d)$. The overall bandwidth complexity of Phase 4 is therefore also upper-bounded by $O(m^3(\log n + b)/d)$.
\end{proof}


\begin{proposition}[Bandwidth complexity of the tunable algorithm]
The bandwidth complexity of the tunable algorithm is $O(m^3(b+\log n)/d)$ and $O(nm(b+\log n)/d)$.
\end{proposition}
\begin{proof}
The proof follows immediately from Lemmas \ref{lemma:hybrid1ub}, \ref{lemma:hybrid2ub}, \ref{lemma:hybrid3ub} and \ref{lemma:hybrid4ub}.
\end{proof}

\section{Conclusions}
\label{sec:conclusions}
In this paper we study the bandwidth complexity of the consensus problem on networks with omnidirectional communication, with particular attention to \emph{robotic} applications. We provide a novel  definition of bandwidth complexity which captures the bandwidth use of multi-agent systems with  modern Media Access Control mechanisms. 
This definition allows us to show that, even in network of moderate size, bandwidth use can be a limiting factor on the time performance and on the scalability of common consensus algorithms such as flooding.

We then prove a lower bound on the bandwidth complexity of the consensus problem that becomes tight for hierarchically computable consensus functions and we provide a matching bandwidth-optimal algorithm. 
Finally, we extend our previous results in \cite{FR-MP:13}, proving that the hybrid algorithm presented in the paper achieves intermediate bandwidth complexity between the lower bound and the bandwidth complexity of the time-optimal algorithm, according to a user-defined tuning parameter. The tradeoff between worst-case time performance, byte performance and bandwidth performance is shown in Figure \ref{fig:bounds}. The implication of this result is that the hybrid algorithm can be used to achieve \emph{mixed} performance metrics, trading time performance and robustness for byte complexity (representative of energy consumption for communication) and bandwidth complexity.

We conclude this paper with a discussion of the limitations of our analysis, which reflect in interesting directions for future research. First, our worst-case analysis provides lower bounds on bandwidth complexity for the class $\mathcal{G}$ of \emph{all} graphs with $n$ nodes: it is of interest to further refine our results to capture the effect of network topology (and in particular of the maximum node degree) on the fundamental limitations on bandwidth performance and on the performance of existing algorithms. Second, an average-case analysis over graphs and asynchronous executions drawn randomly from a representative probability distribution would provide significant insight into the effect of bandwidth complexity on real-world networks, where a (small) probability of collisions is acceptable in presence of a TCP mechanism. Third, while the class of hierarchically computable functions encompasses many relevant engineering problems, it is of interest to study the role of bandwidth complexity for consensus functions that are \emph{not} hierarchically computable. In particular, nonhierarchical algorithms such as average-based consensus may perform no worse than hierarchical algorithms such as GHS with convergecast when the consensus function does not benefit from hierarchical computation. Finally, accurate software and hardware simulations of the performance of the algorithms presented in this paper on wireless channels with different MAC mechanisms will provide further insight into the relevance of our bandwidth metric and into the relative benefits of the different approaches.

\bibliographystyle{IEEEtran}
\bibliography{../../../bib/alias,../../../bib/main}

\end{document}